\documentclass[12pt]{article}

\usepackage[utf8]{inputenc}
\usepackage[margin=1in]{geometry}
\usepackage{setspace}
\usepackage{amsmath,amssymb,amsthm}
\usepackage{graphicx}
\graphicspath{{figures/}}
\usepackage{booktabs}
\usepackage{multirow}
\usepackage{longtable}
\usepackage{caption}
\usepackage{subcaption}

\usepackage{natbib}
\usepackage[colorlinks=true, citecolor=blue]{hyperref}
\usepackage{rotating}
\usepackage{lscape}
\usepackage{array}
\usepackage{threeparttable}
\usepackage{dcolumn}

\usepackage{float}

\hypersetup{hidelinks}
\usepackage{needspace}
\newcolumntype{d}[1]{D{.}{.}{#1}}

\newtheorem{proposition}{Proposition}

\newtheorem{lemma}{Lemma}

\begin{document}

\title{\textbf{The Virtue of Sparsity in Complexity}}

\author{
Nima Afsharhajari, Jonathan Yu-Meng Li
\\
\small Telfer School of Management\\
\small University of Ottawa\\
\small \texttt{nafsh100@uottawa.ca, jonathan.li@telfer.uottawa.ca}
}

\date{This Version: September 14, 2026\\[0.1cm]\small}

\hypersetup{pageanchor=false}
\maketitle
\thispagestyle{empty}

\begin{abstract}
\begingroup
Sparsity or complexity? In modern high-dimensional asset pricing, these are often viewed as competing principles: recent empirical evidence favors richer models, while economic intuition has long favored parsimony. We reconcile this tension by distinguishing \emph{capacity sparsity}---restrictions on effective model capacity---from \emph{factor sparsity}---the parsimonious structure of priced risks. Revisiting the benchmark empirical design of \citet{DidisheimKeKellyMalamud2025a}, we combine nonlinear feature expansions with basis pursuit, using column generation and GPU acceleration to scale estimation to 432 million candidate factors. Reaching this scale reveals a reversal in out-of-sample performance: sparse portfolios trail dense ridgeless benchmarks at lower complexity but achieve a higher Sharpe ratio and lower pricing error at the largest candidate set. Capacity expansion and factor sparsity are therefore complements: enlarging the candidate space allows a parsimonious pricing kernel to outperform its dense counterpart.
\par\endgroup


\end{abstract}

\newpage
\setcounter{page}{1}
\hypersetup{pageanchor=true}

\section{Introduction}

Recent work \citep{KellyMalamudZhou2024a, DidisheimKeKellyMalamud2025a, KellyMalamud2025a} highlights a \emph{virtue of complexity} in empirical asset pricing: large nonlinear factor sets, when combined with shrinkage, can outperform traditional parsimonious models. As high-dimensional methods become central to machine-learning-based SDF estimation, sparsity---long a guiding principle in empirical asset pricing and econometrics---appears increasingly marginalized. The prevailing message is that richer models dominate simpler ones.

\cite{DidisheimKeKellyMalamud2025a} illustrate this shift by comparing models built from large factor sets with (i) traditional low-dimensional factor models and (ii) models that compress a large set of factors into a few principal components \citep{KozakNagelSantosh2020a}. They show that high-dimensional specifications outperform both alternatives and argue that principal-component compression may discard low-variance directions that still contain economically meaningful pricing information. Yet the idea that ``more is better" remains difficult to reconcile with the long-standing economic discipline of parsimony in asset pricing. While large factor sets may improve performance, it remains unclear what economic structure an ever-expanding set of factors actually captures.

We reconcile this tension by distinguishing two notions of sparsity: \emph{capacity sparsity} and \emph{factor sparsity}. The former concerns restrictions on effective model capacity, whereas the latter refers to the parsimonious economic structure of priced risks. In \cite{DidisheimKeKellyMalamud2025a}, the sparsity at issue pertains primarily to capacity sparsity, as their critique is directed at the underperformance of approaches that restrict effective model capacity, whether through classical low-dimensional asset-pricing structures, dimension reduction such as PCA, or strong shrinkage imposed on a rich set of candidate factors. Our key insight is that capacity expansion and factor sparsity are complements. Complexity expands the menu of candidate components, while sparsity determines which components are selected. A richer candidate space can therefore improve the performance of a pricing kernel with few active factors.

To operationalize this idea, we estimate conditional stochastic discount factors (SDFs) using a rich nonlinear feature expansion together with basis pursuit, the minimum-$\ell_1$ interpolating solution. The feature expansion enlarges the candidate factor span, while basis pursuit selects a parsimonious pricing kernel within that span. In contrast to approaches that discipline complexity by compressing the candidate space itself, our approach retains high capacity and imposes parsimony at the level of the selected pricing kernel.
 
To examine the empirical value of factor sparsity, we adopt the ``apples-to-apples'' empirical design of \citet{DidisheimKeKellyMalamud2025a}. Holding the information set fixed, we vary model complexity by expanding the number of nonlinear factors $P$. A key property of basis pursuit is that it admits an optimal solution with at most $T$ active factors, where $T$ denotes the number of in-sample observations. Consequently, even as the candidate feature space becomes highly overparameterized---particularly in regimes where $P \gg T$---the resulting SDF remains sparse relative to the large factor models considered in \citet{DidisheimKeKellyMalamud2025a}.

\begingroup
Our main specification uses a rolling training window of $T=180$ months and expands the candidate set to 432 million factors, or a complexity ratio of $c=P/T=2.4$ million. Basis pursuit admits an optimal solution with at most 180 active factors. We make estimation feasible at this scale by combining column generation with GPU acceleration. At lower complexity, the dense SDF outperforms the sparse SDF. As complexity increases, dense performance largely plateaus, while sparse performance improves overall. At the largest candidate set, the annualized Sharpe ratios are approximately $3.59$ for the sparse SDF and $3.53$ for the dense benchmark, and the sparse SDF also attains a lower pricing error. The endpoint advantage is modest, but the reversal in ranking is consequential: it occurs well beyond the upper complexity ratio of $c=1{,}000$ examined by \citet{DidisheimKeKellyMalamud2025a}. A comparison confined to that range would miss the performance attainable by a sparse pricing kernel selected from a much larger candidate space.
\par\endgroup

\begingroup
The return decomposition helps explain this reversal. At high complexity, the dense SDF's mean return and volatility largely stabilize, while the sparse SDF benefits from both higher mean returns and further volatility reduction as the candidate set expands. At the largest candidate set, the sparse portfolio earns a higher mean return and a higher Sharpe ratio despite remaining somewhat more volatile than the dense benchmark. It nearly dominates across the four reported return quantiles: its 25th, 75th, and 90th percentiles are higher, while its 10th percentile remains slightly lower. After scaling both portfolios to a common annualized volatility of 10\%, the sparse SDF also delivers a higher certainty-equivalent return at the largest candidate set for each of the three risk-aversion coefficients examined, $\gamma\in\{1,2,5\}$. The advantage of sparse selection thus survives a utility-based comparison at common volatility.
\par\endgroup

\begingroup
These results reframe the virtue-of-complexity debate. A large candidate factor space can support an effective pricing kernel with few active factors. Factor sparsity therefore remains an economically useful principle in high-capacity asset pricing.
\par\endgroup

\subsection{Related Literature}
\label{sec:lit}
Our paper contributes to the growing literature that integrates high-capacity machine learning methods into empirical asset pricing. Within this literature, two closely related streams have emerged. The first focuses on improving return prediction from large sets of firm characteristics using flexible nonlinear prediction rules and extensive model comparisons. This work primarily addresses the forecasting problem at the level of individual assets or characteristic-sorted portfolios (e.g., \citealp{GuKellyXiu2020a}; see also \citealp{FengGiglioXiu2020a} for evidence on disciplined factor selection in large factor collections). A second stream applies machine learning methods to estimate conditional stochastic discount factors (SDFs) or conditional factor structures that price the cross-section of assets. Representative examples include deep-learning approaches that impose no-arbitrage moment restrictions in the estimation of the pricing kernel (e.g., \citealp{ChenPelgerZhu2024a}) and latent-factor models with conditioning information embedded through modern dimension-reduction architectures (e.g., \citealp{KellyPruittSu2019a,GuKellyXiu2021a}). Our paper falls within this second strand. We study conditional SDF estimation in high-capacity settings and examine how model complexity translates into economic performance.

Our work also contributes to the recent debate on the \emph{virtue of complexity} in empirical asset pricing. Recent evidence suggests that increasingly complex models---including overparameterized feature expansions paired with shrinkage---can deliver large gains in out-of-sample performance in return prediction and in factor pricing models (e.g., \citealp{KellyMalamudZhou2024a,DidisheimKeKellyMalamud2025a}). At the same time, recent critiques argue that some apparent gains from complexity in return prediction can reflect statistical artifacts or unstable extrapolation in highly flexible specifications (e.g., \citealp{Nagel2025a}). Related work clarifies why such critiques need not overturn the core virtue-of-complexity evidence, while also highlighting that complex models can exhibit nuanced behavior that depends on design and estimator choice (e.g., \citealp{KellyMalamud2025a}). Our results complement this debate by sharpening the mechanism: when the hypothesis class is rich, economic performance depends not only on the availability of many candidate components, but also on the selection principle that maps complexity into factor exposures.

More broadly, our paper relates to the long-standing role of sparsity and parsimony in empirical asset pricing. Classical practice emphasizes small factor representations, while modern high-dimensional approaches often discipline estimation through shrinkage and eigen-structure-aware regularization. For example, \citet{KozakNagelSantosh2020a} construct robust SDFs in large candidate sets by shrinking exposures associated with low-variance principal-component directions. In contrast, recent virtue-of-complexity evidence argues that aggressively compressing large factor sets through dimension reduction can discard economically meaningful low-variance directions and thereby degrade pricing performance (e.g., \citealp{DidisheimKeKellyMalamud2025a}). Our analysis clarifies that these perspectives concern different notions of sparsity. We distinguish \emph{capacity sparsity}, which restricts effective model capacity, from \emph{factor sparsity}, which concerns whether the underlying economic structure of priced risks remains parsimonious. This distinction allows complexity and sparsity to coexist: complexity expands the hypothesis class, while sparsity governs selection of economically meaningful components within it.

This perspective also connects to insights from statistical learning theory that emphasize the role of inductive bias in overparameterized regimes. When many solutions fit the data similarly well, generalization is governed not only by model capacity but also by the estimator's implicit preferences over the solution set. \citet{Wilson2025a} develops this idea through the lens of \emph{soft inductive bias}, and related work on interpolation, double descent, and benign overfitting shows how selection rules within high-capacity classes can produce counterintuitive out-of-sample behavior even when training error is driven to zero (e.g., \citealp{HastieMontanariRossetTibshirani2022a,BartlettLongLugosiTsigler2020a}). Applied to asset pricing, the relevant question is therefore not only whether the SDF is estimated in a high-dimensional space, but which pricing kernel is selected from a large set of near-equivalent high-dimensional solutions.

Finally, our findings provide a complementary perspective on the ``factor zoo'' in empirical asset pricing. The proliferation of proposed factors has long been noted (e.g., \citealp{Cochrane2011a}) and is closely tied to concerns about multiple testing, replication, and post-publication decay in characteristic-based predictors (e.g., \citealp{HarveyLiuZhu2016a,HouXueZhang2020a,McLeanPontiff2016a}). At the same time, modern open-data efforts document broad reproducibility of many published predictors when construction choices are standardized (e.g., \citealp{ChenZimmermann2022a, JensenKellyPedersen2023a}). Our results suggest an additional interpretation that is especially salient in high-capacity conditional SDF settings: the challenge is not merely that many candidate factors exist, but that identifying the economically relevant ones requires a powerful search space combined with an effective selection mechanism.

\section{Empirical Methodology}
This section develops the empirical framework for estimating conditional stochastic discount factors (SDFs) and introduces the two notions of sparsity that arise in high-capacity asset pricing models. We begin with the conditional SDF representation used in modern empirical asset pricing. We then distinguish capacity sparsity from factor sparsity and explain how increasing model complexity can help reveal sparse economic structure. Next, we present our estimation approach, which combines scalable nonlinear feature expansions with basis pursuit to identify sparse pricing kernels within a large factor span. Finally, we characterize sparse and dense interpolating SDFs, clarifying when and why sparse selection can outperform ridgeless estimation.

\subsection{Conditional SDF Representation}
Let $R_{t+1} \in \mathbb{R}^{N_t}$ denote the vector of excess returns on the $N_t$ traded assets available at time $t$, and let $\mathcal F_t$ denote the information set available at time $t$. Absence of arbitrage implies the existence of a stochastic discount factor (SDF) $M_{t+1}$ satisfying
\begin{equation}
E[M_{t+1}R_{t+1}|\mathcal F_t]=0_{N_t},
\end{equation}
where $0_{N_t}$ denotes the zero vector in $\mathbb{R}^{N_t}$.

We focus on tradable SDFs. Following \citet{HansenRichard1987a}, we use the normalized managed-portfolio representation
\begin{equation} \label{m_t}
M_{t+1}=1-w_t'R_{t+1},
\end{equation}
where $w_t \in \mathbb{R}^{N_t}$ is $\mathcal F_t$-measurable.

Let $Z_{i,t} \in \mathbb{R}^D$ denote the characteristic vector of asset $i$ at time $t$, and let $s_1,...,s_P$ be $P$ basis functions of these characteristics. Define the feature matrix
$$S_t \in \mathbb{R}^{N_t \times P},\;\;
(S_t)_{i,p} = s_p(Z_{i,t}).
$$
We parameterize portfolio weights linearly in these features:
\begin{equation} \label{w_t}
w_{t}=\frac{1}{\sqrt{N_t}}S_t\lambda,
\end{equation}
where $\lambda \in \mathbb{R}^P$ is a coefficient vector.

This induces the $P$-dimensional vector of characteristic-managed factor returns
\begin{equation}
F_{t+1}:=\frac{1}{\sqrt{N_t}}S_t'R_{t+1} \in \mathbb{R}^P. \label{ft}
\end{equation}
Substituting \eqref{w_t} into \eqref{m_t} gives 
\begin{equation}
M_{t+1}=1-\lambda'F_{t+1}.
\end{equation}
Thus conditional SDF estimation can be viewed as selecting a portfolio within the span of characteristic-managed factors $F_{t+1}$.

Taking unconditional expectations in the pricing condition yields the necessary moment restriction
\begin{equation}
E[(1-\lambda'F_{t+1})F_{t+1}]=0_{P}.
\end{equation}
Let $\Sigma := E[F_{t+1}F_{t+1}']$ and $\mu := E[F_{t+1}]$. When $\Sigma$ is nonsingular, the population coefficient vector satisfies
\begin{equation}
\Sigma \lambda^\ast = \mu.
\end{equation}

Equivalently, $\lambda^*$ solves the first and third problems below and, when $\mu\ne0$, gives an optimal portfolio direction for the second:
\begin{align}
\max_{\lambda} \quad & \mu'\lambda - \tfrac{1}{2}\lambda'\Sigma\lambda,  \\
\max_{\lambda\neq 0} \quad & \frac{\mu'\lambda}{\sqrt{\lambda'\Sigma\lambda}},  \\
\min_{\lambda} \quad & (\mu - \Sigma\lambda)'\Sigma^{-1}(\mu - \Sigma\lambda). 
\end{align}
Empirically, we estimate conditional SDFs through the equivalent least-squares formulation and evaluate performance using Sharpe ratios and pricing errors.\footnote{Because $\Sigma$ is a second-moment matrix, the middle objective equals $S/\sqrt{1+S^2}$ for a portfolio with positive return variance and conventional Sharpe ratio $S$. This is a strictly increasing transformation of $S$. The second objective determines portfolio direction; the first and third also fix its normalization.}

\subsection{Two Notions of Sparsity}\label{sec:two_notions}
Recent work on the ``virtue of complexity'' \citep{DidisheimKeKellyMalamud2025a} studies regimes in which the number of candidate factors \(P\) is large relative to the time-series sample size \(T\). In such settings, estimation is often stabilized through ridge regularization,

\begin{equation}
\max_{\lambda}
\left\{
\mu'\lambda
- \tfrac{1}{2}\lambda'\Sigma\lambda
- \gamma\|\lambda\|_2^2
\right\}.
\end{equation}

For \(\gamma>0\), this formulation introduces \(\ell_2\)-shrinkage into the estimator. For non-negligible values of \(\gamma\), the shrinkage acts as a form of capacity control by disciplining estimation in a large candidate factor space. We refer to this broader idea as \emph{capacity sparsity}. It can also arise through ex ante restrictions on the factor span (for example, selecting a small number of economically motivated factors), through dimension-reduction methods such as principal components, or through shrinkage more broadly.

By contrast, our focus is on a different notion of sparsity, which we call \emph{factor sparsity}. Rather than restricting the size of the factor span, factor sparsity concerns the structure of the pricing kernel itself. Under this view, the true SDF may depend on only a small subset of economically meaningful directions, even when the candidate factor space is large. The key difficulty, however, is that this sparse structure is unknown ex ante. Precisely for that reason, increasing model complexity can play a constructive role. By enlarging the space of candidate factors, richer feature representations expand the set of directions in which the true sparse structure may lie. Complexity therefore facilitates discovery, while sparsity governs selection within the enlarged space of candidate pricing kernels. Figure~\ref{fig:factor_sparsity} provides a schematic illustration of these two notions of sparsity and their interplay. In the figure, the sparse region can be understood as the set of pricing kernels with at most \(T\) active loadings, consistent with the later role of basis pursuit when the true SDF is representable in this way.

The entire region in Figure~\ref{fig:factor_sparsity} represents the candidate space of pricing kernels at a given complexity level. The high-pricing-error region contains kernels that fail to price assets well. The low-pricing-error, overfitting region contains kernels that fit the data in sample but generalize poorly, whereas the low-pricing-error, good out-of-sample region contains kernels that both fit well and generalize well. The black dot denotes the true SDF \(\lambda^\ast\). In the right panel, overparameterization enlarges the candidate space, but the different regions need not expand proportionally, consistent with the intuition in \citet{Wilson2025a}. The figure is intended only as intuition, but it conveys the key idea: with limited capacity, the sparse structure is harder to uncover, whereas a richer candidate space increases the chance that economically meaningful sparse directions become identifiable. Basis pursuit, introduced below, selects a minimum-$\ell_1$ kernel among exact interpolators. Whether that selection improves out-of-sample performance is the empirical question we investigate.

\begin{figure}[H]
    \centering
    \includegraphics[width=0.75\textwidth]{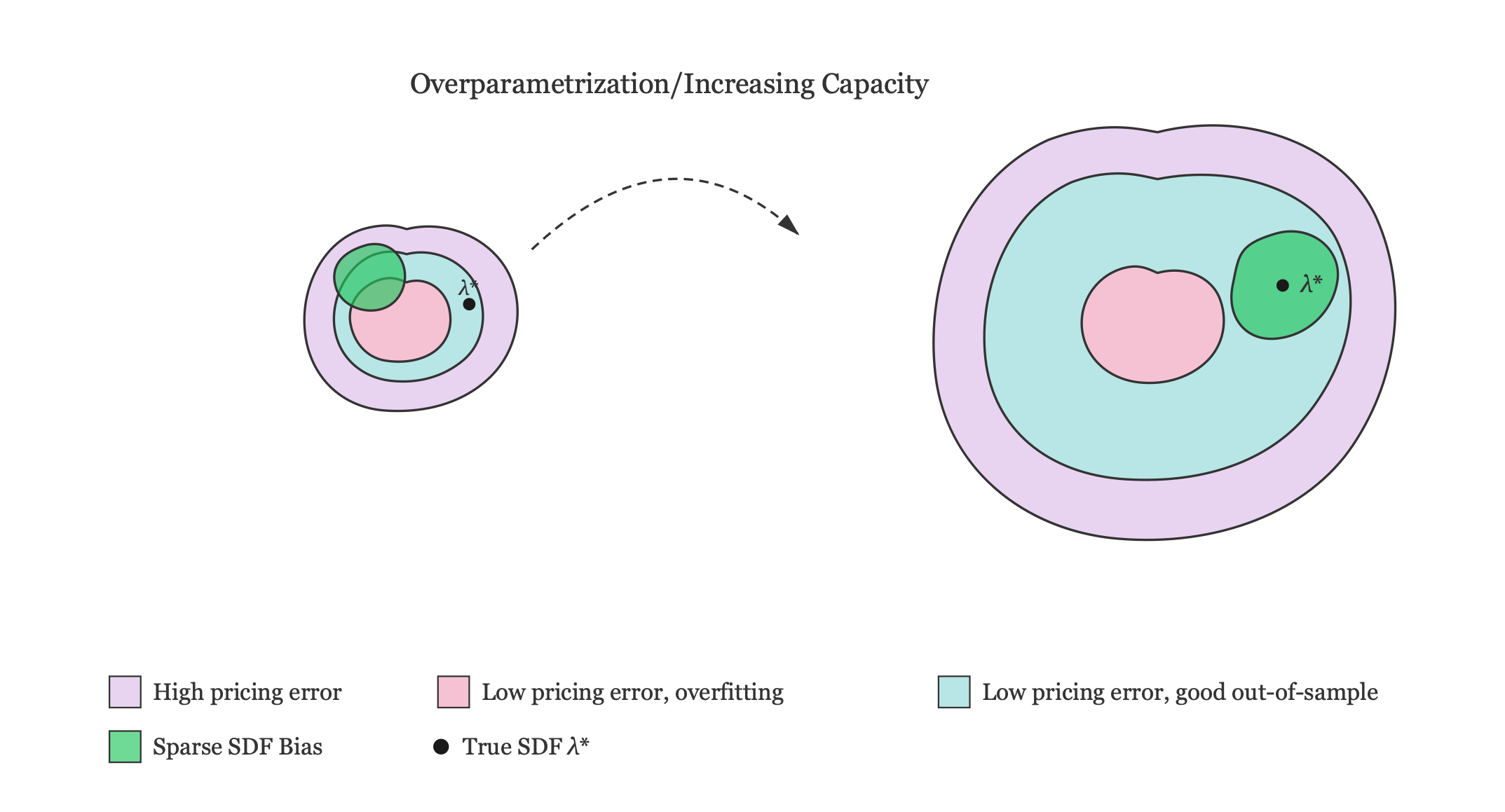}
    \caption{Schematic illustration of the interplay between model capacity and factor sparsity. Left: capacity sparsity restricts the candidate space of pricing kernels, making the true sparse pricing kernel harder to uncover. Right: increasing capacity enlarges the candidate space in which the true sparse pricing kernel may lie, while a sparsity-inducing selection rule favors kernels with at most \(T\) active loadings within that expanded space.}
    \label{fig:factor_sparsity}
  \end{figure}

\subsection{Sparse Selection in Expanded Factor Spaces} \label{sec:factor_sparsity}
To operationalize factor sparsity, we retain a large factor span and impose sparse selection within it. It is convenient to work with the regression form of conditional SDF estimation, because in overparameterized regimes it makes the selection problem especially transparent. Let \(F_P\in\mathbb R^{T\times P}\) denote the in-sample matrix of managed-factor returns at complexity level \(P\), whose \(t\)-th row is \(F_{t+1}'\), where \(F_{t+1}\) is defined in \eqref{ft}. Let \(\mathbf 1\in\mathbb R^T\) denote the unit payoff vector. Let $\mathcal{P}(\lambda)$ denote a regularization term. Then SDF estimation can be written as
\begin{equation} \label{ridge}
\hat{\lambda}
\in
\arg\min_{\lambda}
\frac{1}{2T}\|\mathbf{1}-F_P\lambda\|_2^2
+
\mathcal{P}(\lambda).
\end{equation}
When $P$ is large relative to $T$, the system may admit exact interpolating solutions---many pricing kernels price the managed factors equally well in sample. The central question becomes which SDF is selected from the set of exact interpolating pricing kernels
\[
\mathcal I_P:=\{\lambda\in\mathbb R^P:F_P\lambda=\mathbf 1\}.
\]


A natural selection principle under factor sparsity is the basis-pursuit limit of $\ell_1$ regularization. 

\paragraph{Basis Pursuit}
With $\ell_1$ regularization,
\[
\mathcal{P}(\lambda)=\alpha\|\lambda\|_1,
\]
when exact interpolation is feasible, the $\alpha\downarrow0$ limit of \eqref{ridge} selects a basis-pursuit solution
\begin{equation}
\hat{\lambda}_P^{BP} \in \arg\min_{\lambda \in \mathcal I_P} \|\lambda\|_1.
\end{equation}
This formulation selects the minimum-$\ell_1$ pricing kernel among those that exactly price the managed factors in sample.\footnote{When exact interpolation is infeasible, the vanishing-$\ell_1$-penalty limit selects a minimum-$\ell_1$ solution among least-squares minimizers. If $P<T$ and $F_P$ has full column rank, least squares has a unique solution, so the vanishing-$\ell_1$ and vanishing-$\ell_2$ limits coincide.}

To enrich the candidate factor span without changing the underlying information set, we generate nonlinear transformations of the same $D$ base firm characteristics using random Fourier features (RFF) \citep{RahimiRecht2007a}.
Let \(Z_t\) denote the matrix of standardized characteristics, with row \(i\) the vector \(Z_{i,t} \in \mathbb{R}^{D}\) for stock \(i\). Following \citet{DidisheimKeKellyMalamud2025a}, each random draw generates a sine--cosine pair of features,
\begin{equation}\label{eq:rff}
\begin{aligned}
\bigl[(S_t)_{2p-1}\;\; (S_t)_{2p}\bigr] &= \bigl[\sin(\gamma_p Z_t \boldsymbol{\omega}_p)\;\; \cos(\gamma_p Z_t \boldsymbol{\omega}_p)\bigr], \\
\boldsymbol{\omega}_p &\sim \text{i.i.d.}\ \mathcal{N}(0, \mathbf{I}_{D}), \qquad p = 1, \ldots, P/2,
\end{aligned}
\end{equation}
so that \(P\) features are obtained from \(P/2\) draws. The bandwidth \(\gamma_p\) controls the smoothness of the implied Gaussian kernel; as in \citet{DidisheimKeKellyMalamud2025a} we draw it for each \(\boldsymbol{\omega}_p\) from the grid \(\{0.5, 0.6, 0.7, 0.8, 0.9, 1.0\}\), which embeds varying degrees of nonlinearity in the generated feature set.

\begingroup
\Needspace{8\baselineskip}
\paragraph{Scaling basis pursuit.}
We develop a scalable implementation of basis pursuit for this asset-pricing setting by combining its linear-programming formulation with column generation and GPU acceleration. This allows us to search 432 million candidate factors and reach the complexity range in which the sparse SDF overtakes the dense benchmark in our main specification. Ridgeless estimation reduces to a linear system of dimension $T$ once the candidate factors have been processed. Basis pursuit instead requires optimization over their loadings. We exploit its equivalent linear-programming formulation \citep{ChenDonohoSaunders1998a}: writing $\lambda=u-v$, with $u,v\in\mathbb{R}_+^P$, we solve
\begin{equation}
\begin{aligned}
\min_{u,v\in\mathbb{R}_+^P}\quad & \mathbf{1}_P'(u+v)\\
\text{subject to}\quad & F_P(u-v)=\mathbf{1}_T.
\end{aligned}
\label{eq:bp_lp}
\end{equation}
The objective minimizes the sum of absolute factor loadings, while the $T$ equality constraints preserve exact in-sample interpolation. At our largest candidate set, this direct formulation has 864 million nonnegative variables but only 180 equality constraints. Repeated estimation across rolling windows and random-feature draws compounds the computational burden: a sparse solution can still require an enormous search.

Our strategy is to keep the optimization focused on a working set while preserving access to every candidate factor. We use column generation \citep{BarnhartEtAl1998a} to do so. Starting from a subset that permits exact interpolation, we solve the smaller LP and obtain shadow values for its $T$ equality constraints. We use these values to search the omitted factors for opportunities to improve the objective, add promising candidates, and reoptimize all loadings in the expanded working set. Previously included factors can receive zero loadings, while omitted factors remain eligible for later selection.

The dual provides an optimality check for the full candidate space. Let $y\in\mathbb{R}^T$ be an optimal dual solution of the working LP, representing the equality constraints' shadow values, and let $f_j$ denote column $j$ of $F_P$. An omitted factor with $|f_j'y|>1$ is a candidate for inclusion. If $|f_j'y|\leq1$ for every candidate, the working solution also solves the full LP.\footnote{The full dual is $\max_{y\in\mathbb{R}^T}\mathbf{1}_T'y$ subject to $|f_j'y|\leq1$ for every $j$. Strong duality gives equal objectives for the working primal and dual. If its dual solution is feasible for the full problem, the working primal solution attains the full optimum. This is an exact-arithmetic statement; numerical implementations use feasibility and optimality tolerances.}

The candidate search still spans the entire factor universe. We accelerate this step on the GPU by evaluating the inner products $f_j'y$ across candidates in parallel, and we also use the GPU for feature construction. Column generation reduces the size of the LPs we solve; GPU acceleration makes the searches over the expanded candidate space feasible. We also use GPU acceleration to estimate the dense benchmark. This computational design enables the comparison at the scale where the performance ranking reverses, revealing an advantage of sparse selection that would be missed at lower complexity.
\par\endgroup

\subsection{Sparse versus Dense SDFs}\label{sec:sparse-vs-dense}

In support of our proposed procedure, we characterize when and why the sparse SDF selected by basis pursuit can outperform its dense counterpart given by the ridgeless $\ell_2$ criterion. Consider the limit of the ridge-regularized problem \eqref{ridge} with $\mathcal P(\lambda)=\alpha\|\lambda\|_2^2$ as $\alpha\downarrow0$. In this subsection, assume $F_P$ has full row rank $T<P$, so exact interpolation is feasible. The vanishing-penalty limit selects the minimum-$\ell_2$ interpolator
\begin{equation}
\hat\lambda_P^{RL}\in\arg\min_{\lambda\in\mathcal I_P}\|\lambda\|_2,
\label{eq:ridgeless_interp}
\end{equation}
commonly referred to as ridgeless.

The distinction between sparse and dense SDFs is not one of in-sample fit: both interpolate exactly. It is entirely about \emph{selection} from the same interpolation set. We begin with a basic structural property of basis pursuit: even in the overparameterized regime, it admits an optimal solution whose support size is bounded by \(T\).

\begin{lemma}[Existence of a sparse basis-pursuit solution]
\label{lem:bp_sparse0}
There exists an optimal solution
\[
\hat\lambda_P^{BP}\in\arg\min_{\lambda\in\mathcal I_P}\|\lambda\|_1
\]
whose support size is at most $T$.
\end{lemma}

For our purposes, the key implication of Lemma~\ref{lem:bp_sparse0} is that enlarging the candidate factor span need not lead to a commensurate increase in the effective dimension of the selected SDF. Thus, basis pursuit can search over a much richer nonlinear span while remaining low-dimensional relative to the ambient feature space.

We next characterize how sparse and dense SDFs can differ in out-of-sample mean return. Let
\[
\Pi_P:=F_P'(F_PF_P')^{-1}F_P
\]
denote the orthogonal projector onto \(\operatorname{row}(F_P)\), and let
\[
M_P:=I-\Pi_P
\]
denote the orthogonal projector onto \(\ker(F_P)\). Here \(\operatorname{row}(F_P)\) is the subspace spanned by the rows of the in-sample managed-factor matrix \(F_P\), and \(\ker(F_P)\) is its orthogonal complement.

Let \(F_{P,t+1}^{oos}\in\mathbb R^P\) denote the out-of-sample managed-factor return vector associated with the \(P\)-dimensional feature expansion, and define
\[
\mu_P:=\mathbb E[F_{P,t+1}^{oos}],
\]
where the expectation is taken with respect to the data-generating distribution of firm-characteristic and next-period-return pairs. Write
\[
\mu_P^\parallel:=\Pi_P\mu_P,\qquad
\mu_P^\perp:=M_P\mu_P,
\]
so that
\[
\mu_P=\mu_P^\parallel+\mu_P^\perp,
\qquad
\mu_P^\parallel\in\operatorname{row}(F_P),\qquad
\mu_P^\perp\in\ker(F_P).
\]
This decomposition is natural for comparing sparse and dense SDFs. As shown below, the distinction between \(\operatorname{row}(F_P)\) and \(\ker(F_P)\) governs exactly how sparse and dense selection can differ in mean return.

Define
\[
\bar r_P(\lambda):=\mu_P'\lambda.
\]

\begin{proposition}[Mean-return gap and the role of complexity]
\label{prop:gap_identity1}
Suppose $F_P\in\mathbb R^{T\times P}$ has rank $T<P$ and $\mathcal I_P\neq\varnothing$. Then the ridgeless interpolator is
\[
\hat\lambda_P^{RL}=F_P'(F_PF_P')^{-1}\mathbf 1,
\]
so $\hat\lambda_P^{RL}\in\operatorname{row}(F_P)$, and every feasible interpolator can be written as
\[
\lambda=\hat\lambda_P^{RL}+h,
\qquad
h\in\ker(F_P).
\]
Let
\[
h_P^{BP}:=\hat\lambda_P^{BP}-\hat\lambda_P^{RL}.
\]
Then $h_P^{BP}\in\ker(F_P)$ and
\begin{equation}
\bar r_P(\hat\lambda_P^{BP})-\bar r_P(\hat\lambda_P^{RL})
=
{\mu_P^\perp}'h_P^{BP}
=
\|\mu_P^\perp\|_2\,\|h_P^{BP}\|_2\cos\theta_P,
\label{eq:mean_gap_identity}
\end{equation}
where $\theta_P$ is the angle between $\mu_P^\perp$ and $h_P^{BP}$ when both vectors are nonzero. If either is zero, the mean-return gap is zero and the final expression is interpreted as zero.

Suppose that over a range of increasing complexity levels $P$:
\begin{enumerate}
\item[(i)] $\|\mu_P^\perp\|_2$ is increasing in $P$,
\item[(ii)] $\|h_P^{BP}\|_2\ge \underline h>0$,
\item[(iii)] $\cos\theta_P\ge \underline\rho>0$.
\end{enumerate}
Then
\[
\bar r_P(\hat\lambda_P^{BP})-\bar r_P(\hat\lambda_P^{RL})
\ge
\underline h\,\underline\rho\,\|\mu_P^\perp\|_2,
\]
so the lower bound on the mean-return gap increases with $P$.
\end{proposition}

Proposition~\ref{prop:gap_identity1} makes the mechanism transparent. The ridgeless interpolator is confined to $\operatorname{row}(F_P)$, whereas any feasible deviation from ridgeless must lie in $\ker(F_P)$. From the perspective of expected returns, this means that the minimum-$\ell_2$ rule can exploit only the component of population mean structure that projects onto the subspace generated by historical in-sample pricing relations; any mean-bearing structure orthogonal to that subspace is necessarily left unused. Basis pursuit, by contrast, is not confined in this way: its kernel displacement can access such directions while preserving exact in-sample fit. The mean-return advantage of basis pursuit is therefore governed by three quantities: the magnitude of omitted mean-bearing structure $\|\mu_P^\perp\|_2$, the size of the sparse kernel displacement $\|h_P^{BP}\|_2$, and its alignment with that structure through $\cos\theta_P$. In particular, basis pursuit exceeds ridgeless in mean exactly when its kernel displacement is positively aligned with the omitted component.

\begingroup
The proposition clarifies how complexity can help under three conditions. Condition $(i)$ requires the richer candidate space to contain a growing component of mean returns outside the ridgeless row space. Condition $(ii)$ requires the basis-pursuit displacement from ridgeless to remain bounded away from zero, and condition $(iii)$ requires that displacement to remain favorably aligned with the omitted mean component. Together, they yield an increasing lower bound on the mean-return advantage. The economic content lies in the alignment condition: expanding the candidate space creates scope for improvement, but sparse selection must exploit directions that carry population mean returns. This is where our hypothesis of the virtue of sparsity in complexity enters.
\par\endgroup

We next examine how the volatility of sparse SDFs behaves as complexity increases. A useful benchmark observation is that the scale of the sparse SDF remains controlled. Let
\[
v_P:=\min_{\lambda\in\mathcal I_P}\|\lambda\|_1.
\]
Then
\[
\|\hat\lambda_P^{BP}\|_2\le \|\hat\lambda_P^{BP}\|_1=v_P.
\]
If the feature expansions are nested in $P$ with existing columns unchanged, then $v_P$ is nonincreasing in $P$, since the optimizer at a lower complexity level remains feasible after padding with zeros. Thus, under nested expansions, the $\ell_2$ norm of the sparse SDF is bounded above by a nonincreasing sequence. If, in addition, the out-of-sample managed-factor covariance matrix $V_P:=\operatorname{Var}(F_{P,t+1}^{oos})$ satisfies
\[
\lambda_{\max}(V_P)\le \bar\sigma^2
\]
over the range of interest, then
\[
\sqrt{\hat\lambda_P^{BP\prime}V_P\hat\lambda_P^{BP}}
\le
\bar\sigma\,\|\hat\lambda_P^{BP}\|_2
\le
\bar\sigma\,v_P.
\]

\begingroup
This benchmark argument complements Proposition~\ref{prop:gap_identity1}. Under the stated conditions, the proposition gives an increasing lower bound on the mean advantage, while the scale argument gives a nonincreasing upper bound on volatility. The latter does not imply that portfolio volatility must decline monotonically. Together, the arguments describe how an expanding candidate space can support higher mean returns without an offsetting increase in volatility, providing a possible explanation for the empirical improvement in sparse SDF Sharpe ratios.
\par\endgroup

\begingroup
The economic implication is that exact in-sample fit leaves room for selection to matter out of sample. Lemma~\ref{lem:bp_sparse0} shows that expanding the candidate set need not increase the support bound of the selected kernel; Proposition~\ref{prop:gap_identity1} identifies the alignment needed for a sparse interpolator to earn a higher mean return. The empirical analysis asks whether a sparse representation can realize this advantage as complexity increases.
\par\endgroup

\section{Empirical Design and Results} \label{sec:empirical_results}

This section presents the empirical design and out-of-sample evidence on how sparse and dense SDFs evolve as complexity $P$ increases. 


\subsection{Data Preprocessing and Experimental Setup}
\paragraph{Data}
We use the same data construction and source as \citet{DidisheimKeKellyMalamud2025a}. The underlying dataset is the comprehensive sample of monthly US stock returns and firm characteristics compiled by \citet{JensenKellyPedersen2023a} (JKP); stock-level data and documentation are available at \url{https://jkpfactors.com}. The universe includes NYSE, AMEX, and NASDAQ common stocks (CRSP share codes 10, 11, or 12) excluding nano caps (market cap below the first percentile of NYSE market cap). We reduce the 153 JKP characteristics to \(D=130\) with the fewest missing values, drop stock-months with more than 30\% missing among these 130, and cross-sectionally rank-standardize each characteristic to \([-0.5,0.5]\) following \citet{GuKellyXiu2020a}.

\paragraph{Model setup}

We summarize model complexity by
\[
c := P/T,
\]
where \(T\) denotes the training-window length. In our implementation, we fix \(T = 180\) months (fifteen years of training data) and vary \(P \in \{18, \ldots, 432000000\}\), which yields complexity ratios
\begin{equation}
c \in \{0.1, \ldots, 2400000\}.
\end{equation}
This grid spans the underparameterized regime (\(c < 1\)), the interpolation threshold (\(c = 1\)), and the overparameterized regime (\(c > 1\)). For each value of \(c\), we generate 30 independent random-feature sets using the sine--cosine construction in \eqref{eq:rff}. Each set corresponds to a realization of \(\{\boldsymbol{\omega}_p,\gamma_p\}_{p=1}^{P/2}\), with bandwidths drawn from the grid specified in Section~\ref{sec:factor_sparsity}. We average the results across these sets to reduce Monte Carlo variation due to any particular feature realization.

\paragraph{Evaluation protocol and performance metrics}
Our evaluation design follows the rolling-window framework of \citet{DidisheimKeKellyMalamud2025a}. Fixing a training-window length of \(T=180\) months (fifteen years), we roll the estimation window forward one month at a time. At each formation month \(t\), we estimate \(\hat{\boldsymbol{\lambda}}_t\) using the most recent \(T\) months of training data. Using the characteristics observed at month \(t\) and the asset returns realized in month \(t+1\), we form the corresponding managed-factor return vector \(F_{t+1}\). The realized out-of-sample SDF portfolio return is then
\[
\hat r_{t+1}=\hat{\boldsymbol{\lambda}}_t'F_{t+1}.
\]
Repeating this procedure yields a sequence of \(N=360\) monthly out-of-sample returns from January 1993 through December 2022, denoted by \(\{\hat r_\tau\}_{\tau=1}^N\).

The primary performance metrics are the annualized out-of-sample Sharpe ratio, computed from monthly returns,
\begin{equation}
\widehat{\mathrm{SR}} = \sqrt{12}\,\frac{\bar r}{\hat{\sigma}_r},
\qquad
\bar r = \frac{1}{N}\sum_{\tau=1}^{N}\hat r_\tau,
\qquad
\hat{\sigma}_r = \sqrt{\frac{1}{N-1}\sum_{\tau=1}^{N}(\hat r_\tau-\bar r)^2},
\end{equation}
and the out-of-sample Hansen--Jagannathan distance (HJD),\footnote{We retain the quadratic-form (squared-distance) convention used by \citet{DidisheimKeKellyMalamud2025a}.}
\begin{equation}
\hat{\mathcal D}^{HJ}
=
\hat{\mathbb E}_{OS}\!\left[\hat M_\tau F_\tau\right]'
\hat{\mathbb E}_{OS}\!\left[F_\tau F_\tau'\right]^+
\hat{\mathbb E}_{OS}\!\left[\hat M_\tau F_\tau\right],
\end{equation}
where
\[
\hat M_\tau = 1-\hat{\boldsymbol{\lambda}}_{\tau-1}'F_\tau,
\qquad
\hat{\mathbb E}_{OS}[\,\cdot\,]=\frac{1}{N}\sum_{\tau=1}^{N}[\,\cdot\,],
\]
and \((\cdot)^+\) denotes the Moore--Penrose pseudoinverse. Here \(F_\tau\) denotes the managed-factor excess-return vector in the \(\tau\)-th out-of-sample month.

\begingroup
To characterize the performance comparison more fully, we examine return quantiles alongside the mean and volatility. Let $\{\hat r_\tau(c)\}_{\tau=1}^N$ denote monthly out-of-sample SDF portfolio returns at complexity $c$, and let $Q_q(c)$ denote their empirical $q$-quantile. We report $q\in\{0.10,0.25,0.75,0.90\}$ to show how the downside and upside comparisons change with complexity. \hyperref[app:tail_means]{Appendix C} complements these quantiles with average returns below the two lower cutoffs and above the two upper cutoffs.
\par\endgroup

\begin{sloppypar}
Finally, to translate statistical outperformance into investor welfare, we compute certainty-equivalent (CE) returns for a CRRA investor with relative risk aversion parameter \(\gamma\). For this comparison, we scale each portfolio to an annualized volatility of $10\%$. Let $\tilde r_\tau(c)$ denote the scaled monthly return and $R_\tau(c)=1+\tilde r_\tau(c)>0$ its gross return. For $\gamma\ne1$, the CE gross return $R^{\mathrm{CE}}(c)$ is the constant gross return that delivers the same average power utility as the realized return sequence:
\begin{equation}
\frac{\bigl(R^{\mathrm{CE}}(c)\bigr)^{1-\gamma}}{1-\gamma}
\;=\;
\frac{1}{N}\sum_{\tau=1}^{N} \frac{R_{\tau}(c)^{\,1-\gamma}}{1-\gamma},
\end{equation}
which implies
\begin{equation}
R^{\mathrm{CE}}(c)=\left(\frac{1}{N}\sum_{\tau=1}^{N} R_{\tau}(c)^{\,1-\gamma}\right)^{\frac{1}{1-\gamma}},
\qquad
\mathrm{CE}(c)=R^{\mathrm{CE}}(c)-1.
\end{equation}
\begingroup
For $\gamma=1$, log utility gives
\begin{equation}
R^{\mathrm{CE}}(c)=\exp\!\left(\frac{1}{N}\sum_{\tau=1}^N\log R_\tau(c)\right),
\qquad \mathrm{CE}(c)=R^{\mathrm{CE}}(c)-1.
\end{equation}
\par\endgroup
\end{sloppypar}

\subsection{Main Results: Virtue of Complexity Curves}
Figure \ref{fig:voc_main_panel} presents the main results: virtue-of-complexity curves for out-of-sample mean returns, volatility, Sharpe ratios, and pricing error as functions of complexity \(c\).\footnote{Because the full complexity grid compresses the low-\(c\) region, we report virtue-of-complexity curves for \(c<10\) separately in \hyperref[app:low_c]{Appendix B}.} In the overparameterized regime, the dense SDF (ridgeless) plateaus in mean return, whereas the sparse SDF (basis pursuit) improves overall as complexity rises. Volatility declines for both specifications, with the dense SDF remaining less volatile. Pricing error also falls for both models, although the sparse SDF ultimately achieves the lower pricing error. As a result, the sparse SDF eventually overtakes the dense SDF in Sharpe ratio. 

\begingroup
The extent of the complexity grid is consequential for this comparison. In the main 180-month specification, stopping at $c=1{,}000$, the upper ratio considered by \citet{DidisheimKeKellyMalamud2025a}, would leave the sparse SDF behind the dense benchmark. Extending the grid to $c=2.4$ million changes that ordering: the sparse SDF attains both a higher Sharpe ratio and a lower pricing error at the largest candidate set. Our scalable implementation described in Section~\ref{sec:factor_sparsity} makes this comparison feasible. Its contribution is to reveal a benefit of sparse selection that would be missed by evaluating basis pursuit only at lower complexity.
\par\endgroup

\begin{figure}[H]
    \centering
    \includegraphics[width=\textwidth]{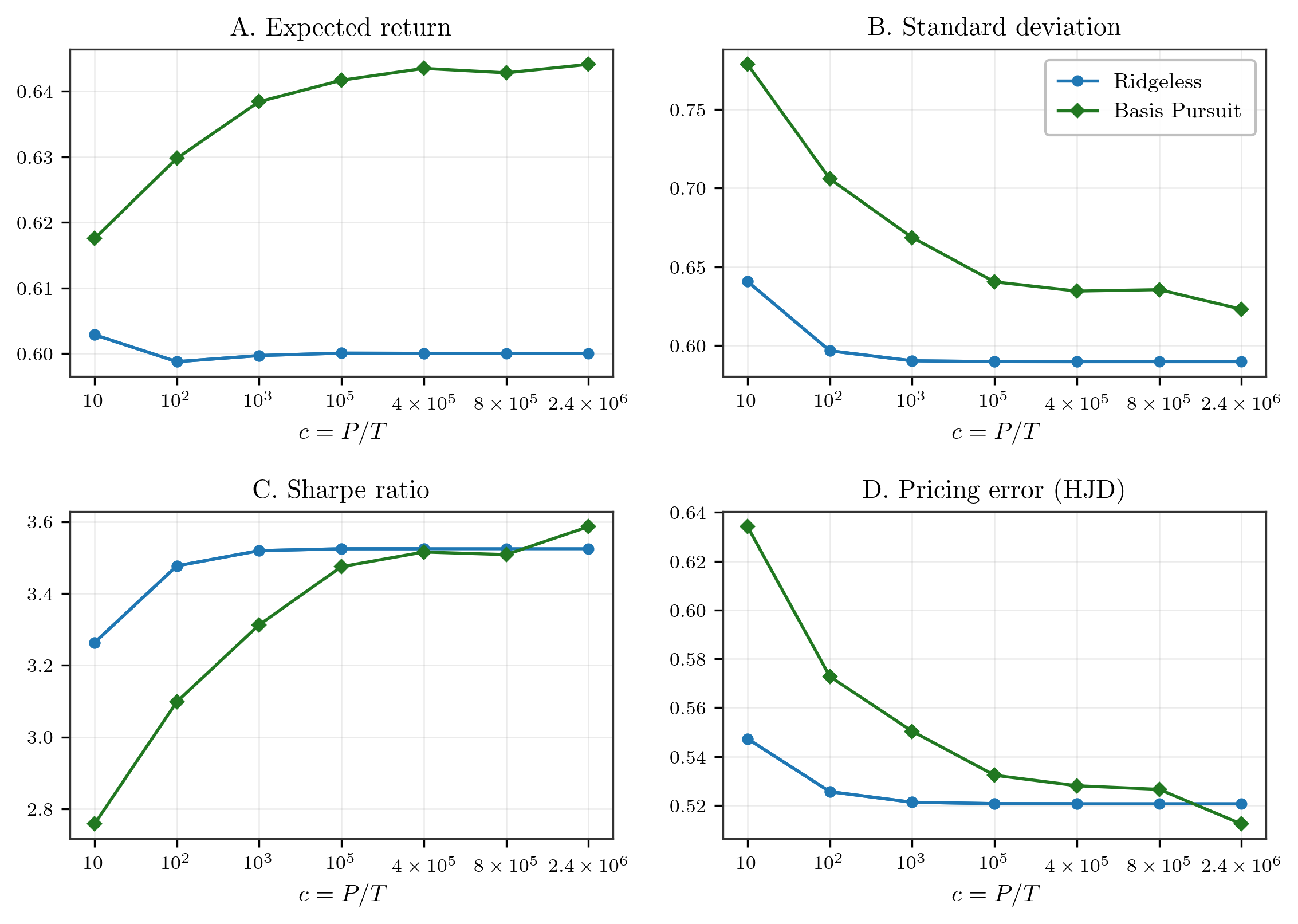}
    \caption{Virtue-of-complexity curves for the sparse SDF (basis pursuit, green) and the dense SDF (ridgeless, blue) as a function of complexity \(c = P/T\). Panel A reports mean returns, Panel B volatility, Panel C annualized Sharpe ratios, and Panel D pricing error (Hansen--Jagannathan distance). Complexity is on an equally spaced axis.}
    \label{fig:voc_main_panel}
  \end{figure}


\subsection{Near-Dominance}\label{sec:quantiles}

\begingroup
Despite its higher volatility, the sparse SDF nearly dominates the dense benchmark across the four reported return quantiles at the largest candidate set. Figure~\ref{distribution-dominance} shows that its 25th, 75th, and 90th percentiles are higher, while its 10th percentile is only slightly lower. The gains therefore extend to part of the downside distribution as well as the upper tail. This pattern emerges as complexity increases: at lower complexity, the sparse SDF already has higher upper-tail returns but lower downside quantiles than the dense benchmark. Expanding the candidate set substantially improves its downside quantiles, narrowing the shortfall at the 10th percentile and reversing it at the 25th.

The sparse SDF's average return in the worst 10\% of months also improves substantially with complexity but remains below that of the dense benchmark. \hyperref[app:tail_means]{Appendix C} reports the complete lower- and upper-tail mean comparison. 
\par\endgroup

Figure~\ref{fig:cumulative_returns} compares cumulative sums of out-of-sample returns after scaling both SDF portfolios to a common volatility. Their paths remain close during the first decade, after which the sparse SDF pulls ahead. The widening cumulative-return gap complements the distributional comparisons above by showing how the performance difference develops over the evaluation period.

\begin{figure}[H]
  \centering
  \includegraphics[width=\textwidth]{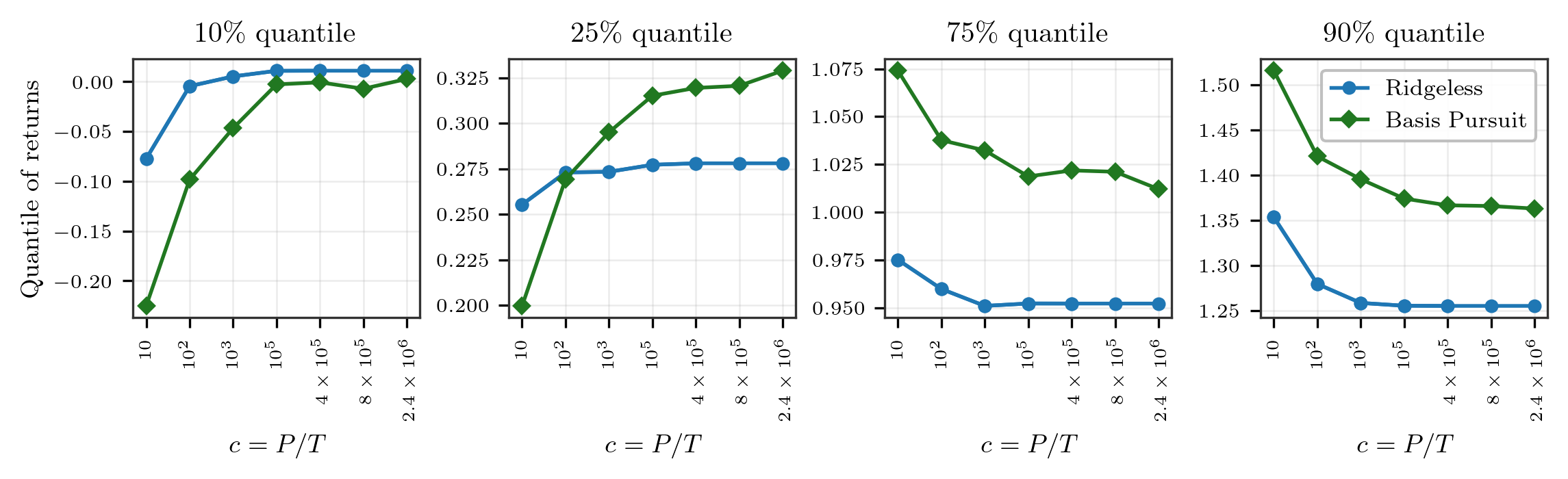}
  \caption{Monthly out-of-sample return quantiles for the sparse SDF (basis pursuit, green) and the dense SDF (ridgeless, blue) across complexity levels. The panels report the 10th, 25th, 75th, and 90th percentiles. At the largest candidate set, the sparse SDF has higher returns at the 25th, 75th, and 90th percentiles and a slightly lower return at the 10th percentile. \hyperref[app:tail_means]{Appendix C} reports the corresponding lower- and upper-tail means.}
  \label{distribution-dominance}
  \label{fig:quantiles_panel}
\end{figure}

\begin{figure}[H]
  \centering
  \includegraphics[width=0.72\textwidth]{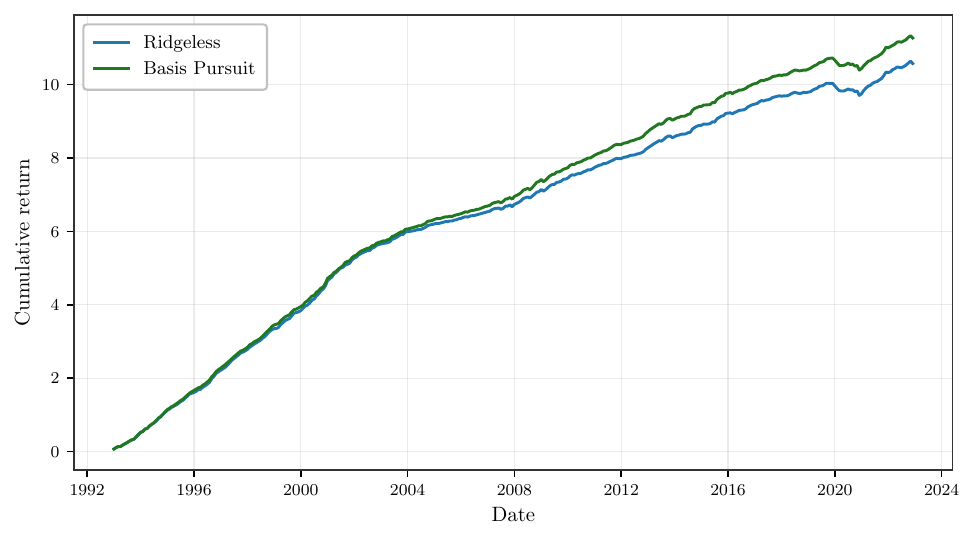}
  \caption{Cumulative sums of monthly out-of-sample returns for the sparse SDF and the dense SDF, after scaling both series to a common volatility.}
  \label{fig:cumulative_returns}
\end{figure}

\subsection{Welfare Gains from Sparsity: Certainty-Equivalent Analysis}\label{sec:welfare}

To quantify the welfare gains delivered by sparse SDFs, we compute certainty-equivalent (CE) returns for a constant relative risk aversion (CRRA) investor. Unlike the Sharpe ratio, the CE maps the full distribution of SDF portfolio returns into investor utility and thus provides a direct measure of economic value.\footnote{Before computing CE returns, we scale both SDF portfolios to a common annualized volatility of \(10\%\). The unscaled portfolios can lose more than their invested capital in some months; such observations fall outside the positive-wealth domain of CRRA utility and prevent a valid CE comparison on the unscaled return series.}

Figure~\ref{fig:ce_gamma_panel} shows that the performance reversal between sparse and dense SDFs extends to investor welfare. At the highest complexity level, the sparse SDF delivers a higher certainty-equivalent return for each of the three risk-aversion coefficients shown, \(\gamma\in\{1,2,5\}\). Greater risk aversion lowers the CE of both portfolios, but does not overturn the sparse SDF's advantage at this complexity level. At common volatility, the advantage of sparse selection thus extends to the three investor preferences examined.

\begin{figure}[H]
  \centering
  \includegraphics[width=\textwidth]{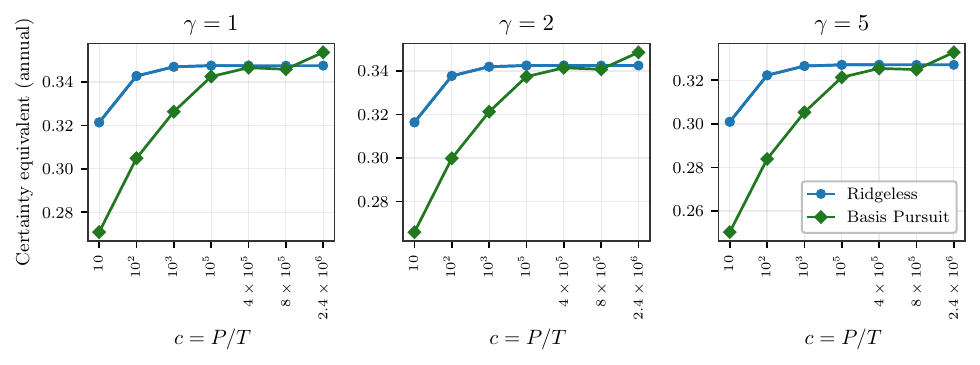}
  \caption{Certainty-equivalent (CE) returns for CRRA investors with $\gamma\in\{1,2,5\}$, after scaling both portfolios to an annualized volatility of $10\%$, reported as annual decimal returns. The sparse SDF delivers the higher CE at the top of the complexity grid at every risk aversion shown.}
  \label{fig:ce_gamma_panel}
\end{figure}

%

%
\subsection{Characteristic Importance: Sparse in Factors, Broad in Economics}\label{sec:char_importance}

The distinction between capacity sparsity and factor sparsity in Section~\ref{sec:two_notions} raises a further empirical question: does a sparse factor representation also concentrate the pricing kernel's sensitivity on a few underlying characteristics? Following \citet{ChenPelgerZhu2024a}, we measure characteristic importance by the average absolute sensitivity of the SDF portfolio weight to each characteristic. Within each random-feature draw, we average sensitivities over estimation windows and normalize across the \(130\) characteristics to obtain shares that sum to one; we then average these shares across draws.

\begin{figure}[H]
  \centering
  \includegraphics[width=\textwidth]{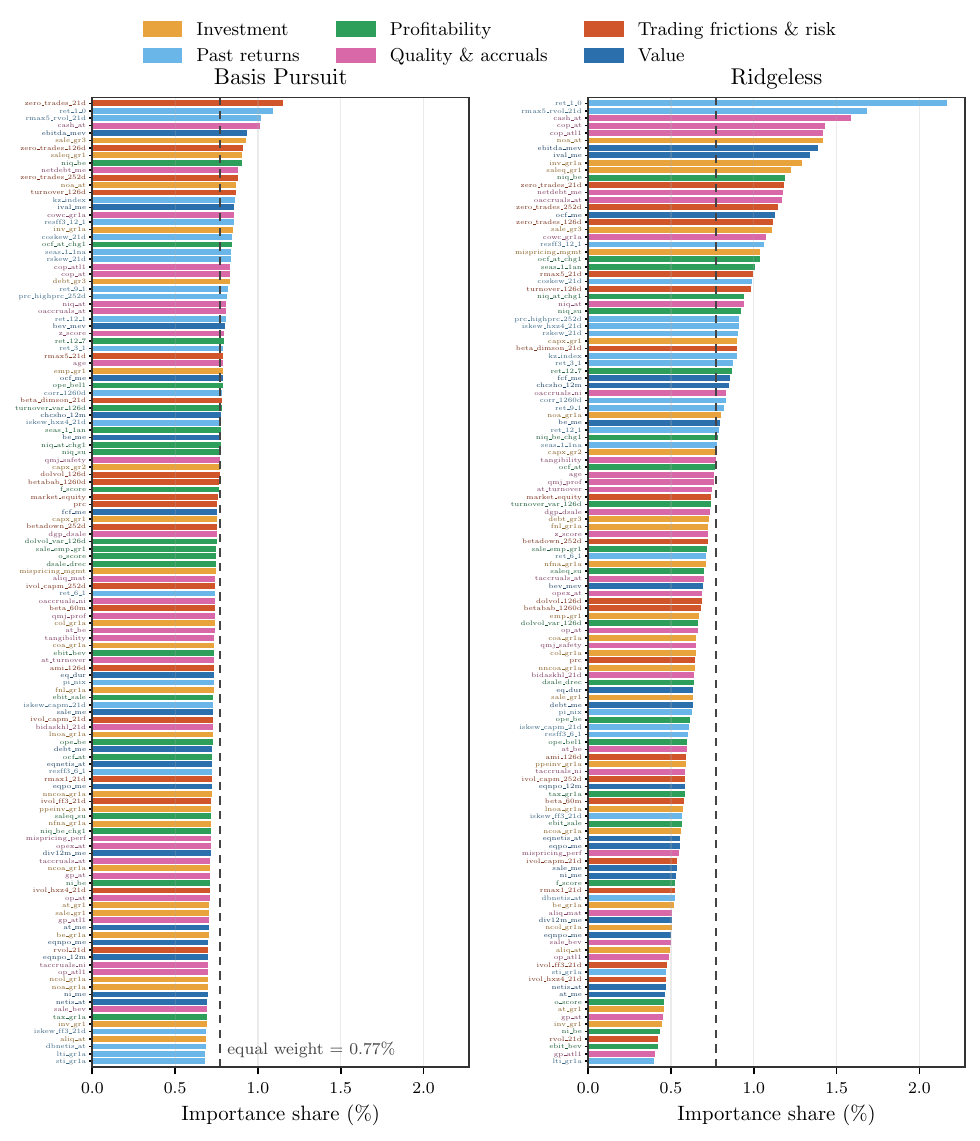}
  \caption{Characteristic importance for all \(130\) characteristics, averaged over 360 out-of-sample months. Bars give each characteristic's share of total sensitivity, in percent; the dashed line marks an equal-share reference of \(1/130 = 0.77\%\). Each panel is ordered by its own method's importance. Color denotes the theme grouping obtained by collapsing the thirteen clusters of \citet{JensenKellyPedersen2023a}.}
  \label{fig:charimp_full}
\end{figure}

\begingroup
Figure~\ref{fig:charimp_full} shows that factor sparsity need not imply concentrated dependence on the underlying characteristics. The dense SDF has the more concentrated sensitivity profile: one-month reversal accounts for $2.17\%$ of total importance, compared with an equal-share reference of $0.77\%$, while the largest share under basis pursuit is $1.15\%$. The difference extends beyond the leading characteristic. In the averaged profiles shown, the sparse SDF's importance shares are more evenly distributed across the full set of 130 inputs.

The two pricing kernels nevertheless share several prominent characteristics. One-month reversal and cash holdings rank among the leading inputs under both methods. The one-month zero-trading measure ranks first under basis pursuit, whereas one-month reversal ranks first under ridgeless estimation. The contrast therefore concerns how the pricing kernels distribute their dependence on familiar characteristics, as well as which characteristics receive the greatest emphasis.

Each random Fourier factor combines the underlying characteristics nonlinearly, so a pricing kernel with few active factors can remain sensitive to a broad set of firm attributes. Both estimators use this feature construction; the figure shows that sparse factor loadings can coexist with a more even distribution of characteristic sensitivity. Selecting few composite factors therefore need not confine the pricing kernel to a narrow set of economic inputs.\footnote{Derivative-based measures of variable importance are developed by \citet{RosascoVillaMosciSantoroVerri2013a}. Following \citet{ChenPelgerZhu2024a}, we use absolute derivatives evaluated out of sample. The normalized averages describe relative sensitivity; they do not identify its sign, compare absolute sensitivity across methods, or establish stability across windows.}
\par\endgroup

\subsection{Sensitivity to Sample Size}\label{sec:tw_robustness}

\begingroup
The training-window length affects both the information available for estimation and the support bound for basis pursuit: an optimal basic solution has at most $T$ active factors. Figure~\ref{fig:tw_robust} therefore asks whether the high-complexity advantage of sparse selection persists when both the estimation sample and this support bound change. We compare rolling windows of $T=12$ and $60$ months with the $T=180$ specification used in the main analysis.
\par\endgroup

\begin{figure}[H]
  \centering
  \includegraphics[width=\textwidth]{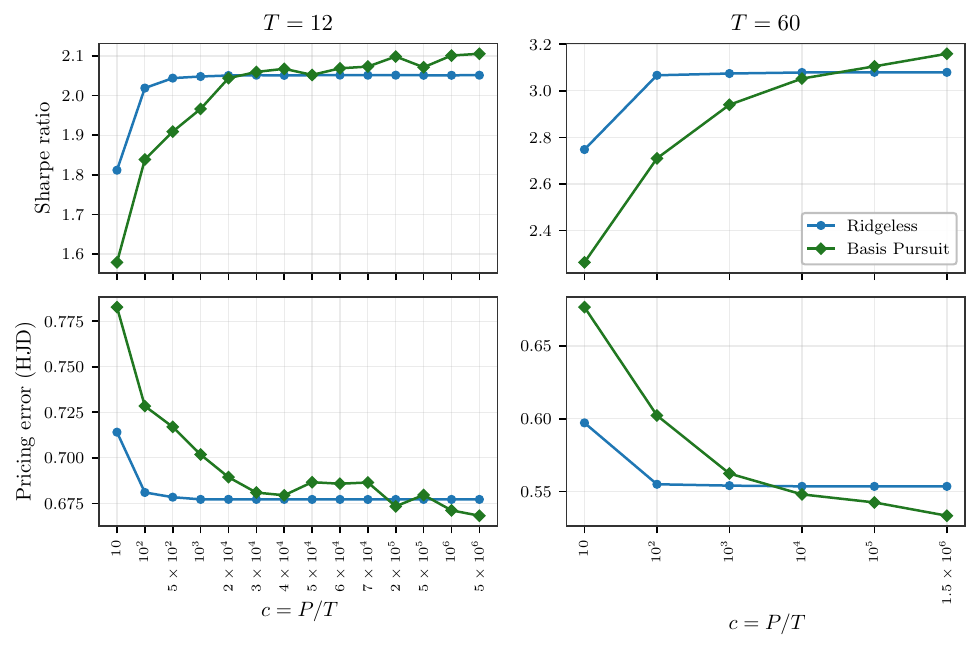}
  \caption{Annualized out-of-sample Sharpe ratio (top row) and pricing error (HJD, bottom row) against complexity for the sparse SDF (basis pursuit, green) and the dense SDF (ridgeless, blue), at training windows of \(T=12\) and \(60\) months. A lower pricing error is better.}
  \label{fig:tw_robust}
\end{figure}

\begingroup
The same qualitative reversal appears in both shorter windows. At lower complexity, the dense SDF performs better. As the candidate set expands, sparse performance improves overall; at the top of each grid, the sparse SDF has a higher Sharpe ratio and a lower pricing error. Its peak annualized Sharpe ratio is approximately $2.1$ for $T=12$ and $3.2$ for $T=60$, compared with $3.6$ in the main analysis. Within each window, these gains are obtained by enlarging the candidate set while the theoretical support bound remains fixed.
\par\endgroup

\begingroup
The crossovers occur at different points on the three complexity grids, and the $T=12$ performance ordering fluctuates at high complexity. A longer estimation sample can enrich the span of observed factor returns available to the ridgeless interpolator, potentially strengthening the dense benchmark; it also raises the support bound for basis pursuit. These changes provide possible explanations for differences across windows, although the figures do not isolate their effects. Moreover, because $c=P/T$, a common complexity ratio across windows corresponds to different numbers of candidate factors. The robust finding is that sparse selection eventually outperforms the dense benchmark in all three specifications, rather than that the reversal occurs at a universal complexity threshold.
\par\endgroup

%
%


\section{Conclusion}

This paper studies how complexity and parsimony interact in empirical asset pricing. We argue that the apparent tension between them reflects a distinction between two different notions of sparsity: \emph{capacity sparsity}, which restricts effective model capacity, and \emph{factor sparsity}, which concerns the parsimonious structure of priced risks within an expanded candidate space. Our results support the view that richer candidate spaces can improve asset-pricing performance by creating more scope for informative sparse selection.

Empirically, we show that when the information set is held fixed and model complexity is increased through richer nonlinear feature expansions, sparse and dense selection produce distinct performance paths. The sparse SDF selected by basis pursuit improves overall in the high-complexity regime, whereas the dense ridgeless benchmark plateaus. Our scalable implementation of basis pursuit reaches the complexity range in which the relative performance of sparse and dense SDFs reverses. The computational extension changes the empirical assessment of sparsity: a comparison confined to lower complexity would miss the performance attainable by a sparse pricing kernel drawn from a much larger candidate space. The sparse SDF nearly dominates across the four reported return quantiles at the largest candidate set, with a small remaining shortfall at the 10th percentile. At a common annualized volatility of $10\%$, it also delivers a higher CE for each of the three risk-aversion coefficients examined. These findings support the economic value of a sparse factor representation.

Our analysis helps clarify why this pattern can arise. In overparameterized settings where many pricing kernels fit the data equally well in sample, performance depends not only on the richness of the factor space but also on the selection rule used to choose among admissible kernels. Sparse selection can outperform dense interpolation when it moves toward economically relevant directions that the minimum-\(\ell_2\) rule does not exploit, and that advantage can become larger as complexity increases.

Taken together, our results reframe the virtue-of-complexity debate. A richer candidate space can improve the performance of a pricing kernel with few active factors. Capacity expansion and factor sparsity can therefore work together in empirical asset pricing.

\appendix
\section*{Appendix A.}

This appendix proves the results in Section~\ref{sec:sparse-vs-dense}. Throughout, let
\[
\mathcal I_P:=\{\lambda\in\mathbb R^P:F_P\lambda=\mathbf 1\},
\]
where \(F_P\in\mathbb R^{T\times P}\) has rank \(T<P\) and \(\mathcal I_P\neq\varnothing\).

\subsection*{A.1 Proof of Lemma~\ref{lem:bp_sparse0}}

\begin{proof}
Write
\[
\lambda=\lambda^+-\lambda^-,
\]
with \(\lambda^+,\lambda^-\in\mathbb R_+^P\). Then the basis-pursuit problem is equivalent to the linear program
\[
\min_{\lambda^+,\lambda^-\ge 0}\ \mathbf 1'(\lambda^++\lambda^-)
\quad\text{s.t.}\quad
F_P(\lambda^+-\lambda^-)=\mathbf 1.
\]
This linear program has \(T\) equality constraints. By the fundamental theorem of linear programming, there exists an optimal basic feasible solution. Any basic feasible solution has at most \(T\) basic variables among the \(2P\) nonnegative variables \((\lambda^+,\lambda^-)\). Therefore at most \(T\) indices can contribute a nonzero entry to
\[
\lambda=\lambda^+-\lambda^-.
\]
Hence there exists an optimal basis-pursuit solution with support size at most \(T\).
\end{proof}

\subsection*{A.2 Proof of Proposition~\ref{prop:gap_identity1}}

\begin{proof}
The ridgeless interpolator solves
\[
\min_{\lambda\in\mathbb R^P}\ \|\lambda\|_2
\quad\text{s.t.}\quad
F_P\lambda=\mathbf 1.
\]
The first-order conditions yield
\[
\hat\lambda_P^{RL}=F_P'(F_PF_P')^{-1}\mathbf 1.
\]
Since \(\hat\lambda_P^{RL}\) is a linear combination of the rows of \(F_P\), it belongs to \(\operatorname{row}(F_P)\).

Now let \(\lambda\in\mathcal I_P\). Since both \(\lambda\) and \(\hat\lambda_P^{RL}\) are feasible,
\[
F_P(\lambda-\hat\lambda_P^{RL})=\mathbf 0,
\]
so \(\lambda-\hat\lambda_P^{RL}\in\ker(F_P)\). Hence every feasible interpolator can be written as
\[
\lambda=\hat\lambda_P^{RL}+h,\qquad h\in\ker(F_P).
\]

By definition of the projection operators,
\[
\mu_P^\parallel=\Pi_P\mu_P,\qquad
\mu_P^\perp=M_P\mu_P,
\]
so
\[
\mu_P=\mu_P^\parallel+\mu_P^\perp,
\qquad
\mu_P^\parallel\in\operatorname{row}(F_P),\qquad
\mu_P^\perp\in\ker(F_P).
\]

Then
\[
\bar r_P(\lambda)-\bar r_P(\hat\lambda_P^{RL})
=
\mu_P'(\lambda-\hat\lambda_P^{RL}).
\]
Because \(\mu_P^\parallel\in\operatorname{row}(F_P)\) and \(\lambda-\hat\lambda_P^{RL}\in\ker(F_P)\), orthogonality implies
\[
{\mu_P^\parallel}'(\lambda-\hat\lambda_P^{RL})=0.
\]
Therefore
\[
\bar r_P(\lambda)-\bar r_P(\hat\lambda_P^{RL})
=
{\mu_P^\perp}'(\lambda-\hat\lambda_P^{RL}).
\]

Applying this to \(\lambda=\hat\lambda_P^{BP}\) gives
\[
\bar r_P(\hat\lambda_P^{BP})-\bar r_P(\hat\lambda_P^{RL})
=
{\mu_P^\perp}'h_P^{BP},
\]
where
\[
h_P^{BP}:=\hat\lambda_P^{BP}-\hat\lambda_P^{RL}\in\ker(F_P).
\]
The cosine representation follows from
\[
{\mu_P^\perp}'h_P^{BP}
=
\|\mu_P^\perp\|_2\,\|h_P^{BP}\|_2\cos\theta_P.
\]

The second part is immediate from \eqref{eq:mean_gap_identity}.

\end{proof}

\section*{Appendix B. Low-\(c\) virtue-of-complexity curves}
\label{app:low_c}

This appendix reports the low-complexity counterparts of the main-text figures, restricted to \(c<10\) so that behavior in the underparameterized and near-interpolation regimes is visible.

\begin{figure}[H]
  \centering
  \includegraphics[width=\textwidth]{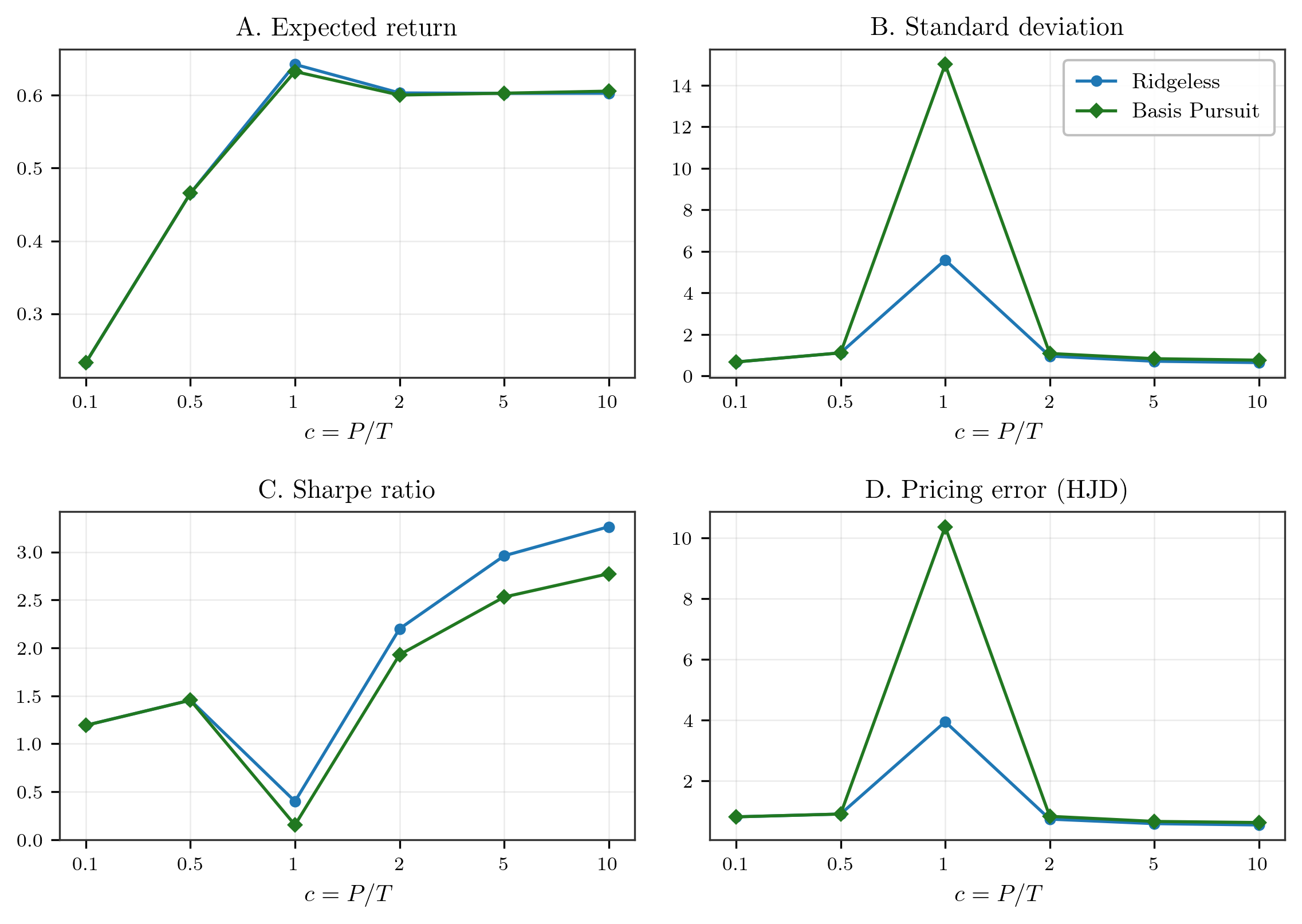}
  \caption{Low-complexity counterpart of Figure~\ref{fig:voc_main_panel}, restricted to \(c<10\). Sharpe ratios are annualized.}
  \label{fig:voc_main_panel_low_c}
\end{figure}

\begin{figure}[H]
  \centering
  \begin{subfigure}{\textwidth}
      \centering
      \includegraphics[width=\linewidth]{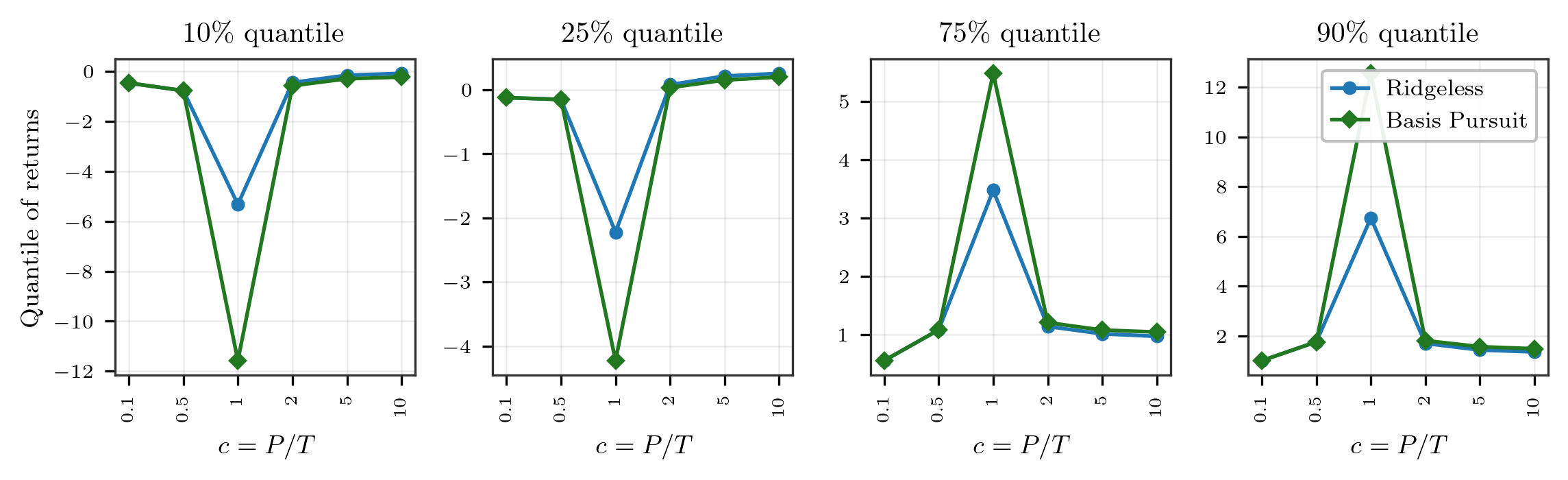}
      \caption{Return quantiles across complexity levels}
      \label{fig:quantiles_panel_low_c}
  \end{subfigure}

  \vspace{0.8em}
  \begin{subfigure}{\textwidth}
      \centering
      \includegraphics[width=\linewidth]{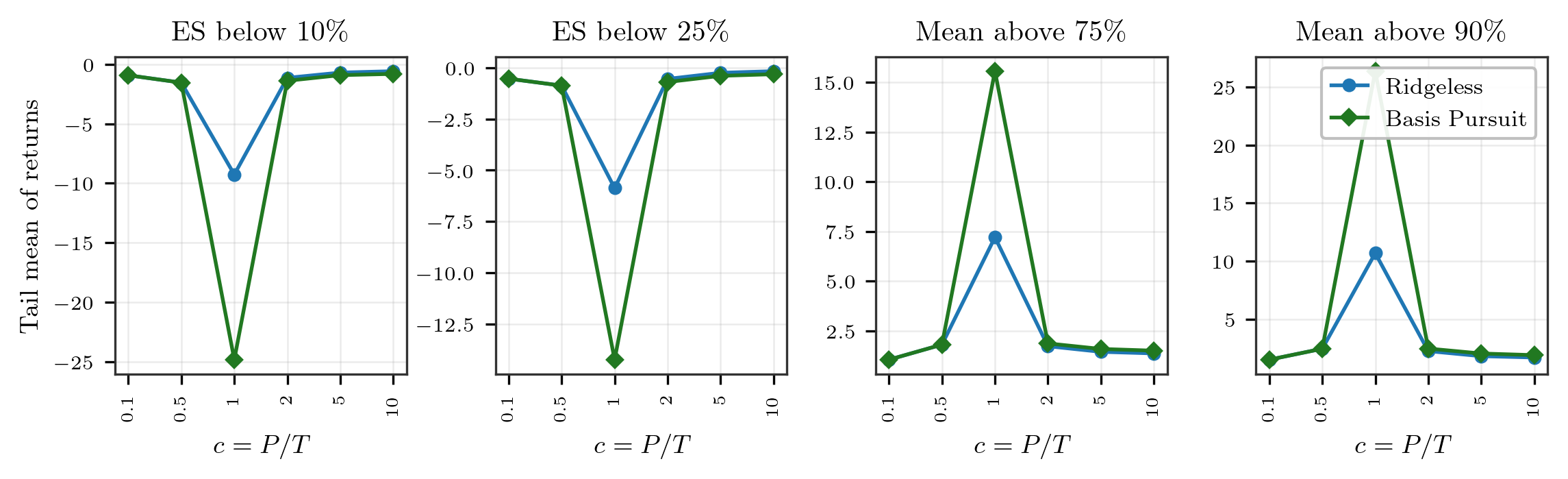}
      \caption{Expected Shortfall (ES) and upper-tail mean across complexity levels}
      \label{fig:cvar_panel_low_c}
  \end{subfigure}

  \caption{Low-complexity counterparts of the quantiles in Figure~\ref{distribution-dominance} and the tail means in Figure~\ref{fig:tailmeans_high_c}, restricted to \(c<10\).}
\label{fig:distribution_dominance_low_c}
\end{figure}

\begin{figure}[H]
  \centering
  \includegraphics[width=\textwidth]{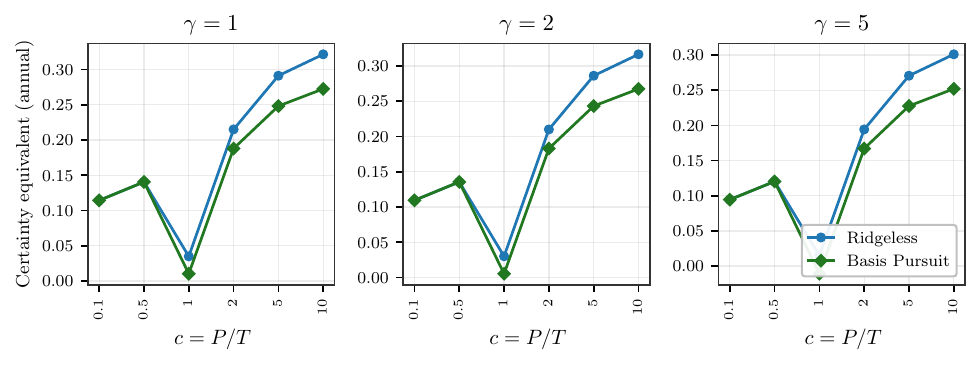}
  \caption{Low-complexity counterpart of Figure~\ref{fig:ce_gamma_panel}, restricted to \(c<10\).}
  \label{fig:ce_gamma_panel_low_c}
\end{figure}

\section*{Appendix C. Lower- and upper-tail mean returns}
\phantomsection\label{app:tail_means}
\begingroup
This appendix complements the return quantiles in Section~\ref{sec:quantiles}. Let $Q_q(c)$ denote the empirical $q$-quantile of monthly out-of-sample returns at complexity $c$. For the lower cutoffs, we report the average return in the corresponding downside tail,
\begin{equation}
\mathrm{ES}_q(c):=
\frac{\sum_{\tau:\hat r_\tau(c)\le Q_q(c)}\hat r_\tau(c)}
{\#\{\tau:\hat r_\tau(c)\le Q_q(c)\}},
\qquad q\in\{0.10,0.25\}.
\end{equation}
For the upper cutoffs, we report
\begin{equation}
\mathrm{UTM}_q(c):=
\frac{\sum_{\tau:\hat r_\tau(c)\ge Q_q(c)}\hat r_\tau(c)}
{\#\{\tau:\hat r_\tau(c)\ge Q_q(c)\}},
\qquad q\in\{0.75,0.90\}.
\end{equation}
ES is expressed as a return, rather than a loss, so a higher value indicates better downside performance. Unlike a quantile cutoff, a tail mean also reflects the severity of outcomes beyond that cutoff.

Figure~\ref{fig:tailmeans_high_c} shows that the sparse SDF's average return in the worst 10\% of months improves substantially with complexity but remains below that of the dense benchmark. Its average return in the worst quartile overtakes the dense benchmark at the largest candidate set, while its upper-tail means remain higher. The remaining worst-decile shortfall qualifies the favorable quantile pattern in the main text.
\par\endgroup

\begin{figure}[H]
  \centering
  \includegraphics[width=\textwidth]{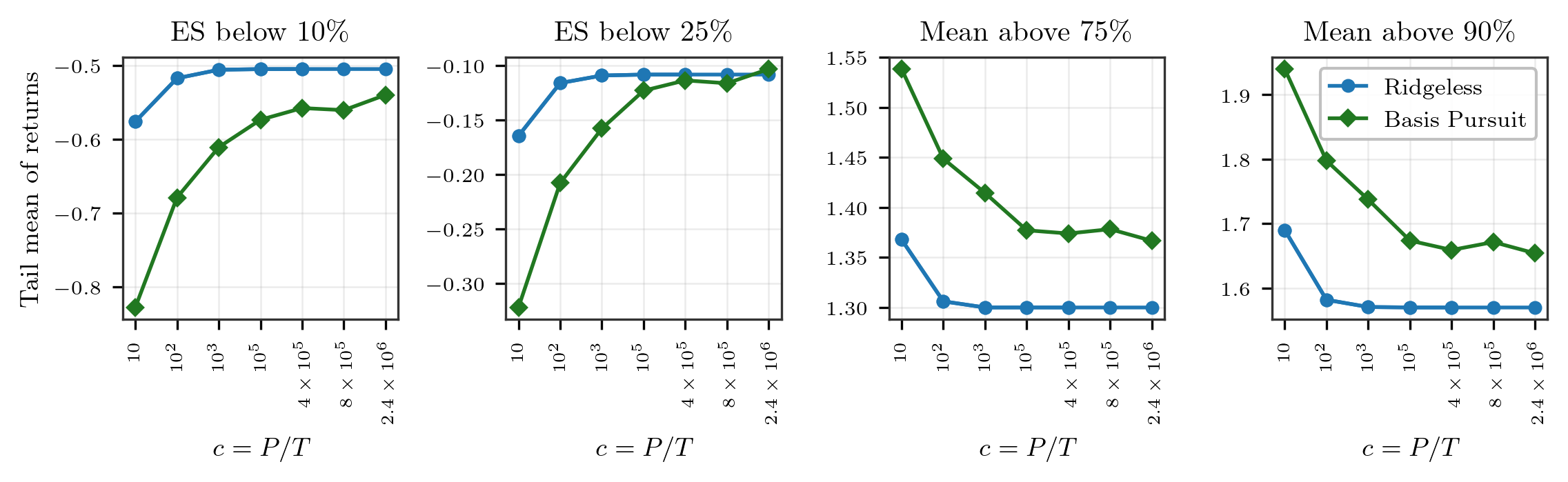}
  \caption{Lower- and upper-tail mean monthly out-of-sample returns for the sparse SDF (basis pursuit, green) and the dense SDF (ridgeless, blue). The panels report average returns below the 10th and 25th percentiles and above the 75th and 90th percentiles.}
  \label{fig:tailmeans_high_c}
  \label{fig:cvar_panel}
\end{figure}

\bibliography{references_ewfc}

\end{document}